\documentclass[final,twoside,11pt]{entics}

\usepackage{cite}
\usepackage{dirtytalk}
\usepackage{enticsmacro}
\usepackage{graphicx}
\usepackage{amsmath,amssymb,tikz-cd}
\usepackage[all]{xy}
\newcommand{\set}{\; \mathsf{set}}
\newcommand{\propa}{\; \mathsf{prop}}
\newcommand{\props}{\; \mathsf{prop_s}}

\newcommand{\mltt}{Martin-Löf's type theory}

\def\conf{MFPS 2023} 	%%Fill in the acronym for your conference (with year)
\def\myconf{mfps39}
\volume{3}			%Fill in the ENTICS volume number here
\def\volu{3}			% and here
\def\papno{15}		%Fill in your paper number here
\def\lastname{Maietti, Sabelli} %Lastnames appear in the running header 
\def\runtitle{Well-founded trees and inductive pointfree topologies} %Short title appears in the running header on even pages. 
\def\mydoi{11755}
\def\copynames{M. Maietti, P. Sabelli}    %% Fill in the first initial and last name of the authors
\begin{document}
%%%Note the beginning and end of the frontmatter section that starts here%%%%%
\begin{frontmatter}
  \title{A Topological Counterpart of  Well-founded Trees in Dependent Type Theory
  %% On the relationship between well-founded tree and inductive generation of pointfree topology
  }
  %%\thanksref{ALL}} 						%%Title here and the
 %%\thanks[ALL]{General thanks to everyone who should be thanked.}   %%Text of \thanks[ALL} here..
 %%%%%%%%%%%%%%%%%%%%%%%%%%%%			This Thanks is optional.
  %%%%Now the author(s) names(s)%%%%%
  \author{Maria Emilia Maietti\thanksref{a}\thanksref{myemail}}	%%Note NO SPACE between 
   \author{Pietro Sabelli\thanksref{a}\thanksref{coemail}}		%last name and \thanksref{...} 
    %%%Next come the addresses%%%%  							
   \thanks[myemail]{Email: \href{mailto:maietti@math.unipd.it} {\texttt{\normalshape
        maietti@math.unipd.it}}} 
   %%%Note: if both authors share same institution, only list the address once, after the second 
   %%%author. 
   %%%There also is a link from the first author to the co-author's address to show how to list 
   %%%affiliations to more than one institution, when needed. 
  \address[a]{Department of Mathematics “Tullio Levi-Civita”\\University of Padua\\
    Padua, Italy} 
  \thanks[coemail]{Email:  \href{mailto:sabelli@math.unipd.it} {\texttt{\normalshape
        sabelli@math.unipd.it}}}
\begin{abstract} 
%% We provide a topological counterpart of well-founded trees (for short W-types) introduced in dependent type theory by using a proof-relevant version of the notion of inductively generated suplattices introduced in the context of formal topology.  In more detail we show that a dependent type calculus endowed with W-types is equivalent to the same calculus endowed with proof-relevant inductively generated basic covers in intensional Martin-Löf's type theory extended either with functions extensionality, or with the definitional $\eta$-equalities for dependent sums and dependent products. This allows to conclude that W-types are equivalent to proof-relevant inductively generated formal covers in Homotopy Type Theory.
Within dependent type theory, we provide a topological counterpart of well-founded trees (for short, W-types) by using a proof-relevant version of the notion of inductively generated suplattices introduced in the context of formal topology
under the name of “inductively generated basic covers”. In more detail, we show, firstly, that in Homotopy Type Theory, W-types and proof-relevant inductively generated basic covers are propositionally mutually encodable. Secondly, we prove they are definitionally mutually encodable in the Agda implementation of intensional Martin-Löf’s type theory. Finally, we reframe the equivalence in the Minimalist Foundation framework by introducing  well-founded predicates as the logical counterpart for predicates of dependent W-types. All the results have been checked in the Agda proof-assistant.
\end{abstract}
\begin{keyword}
dependent type theory, formal topology, well-founded trees, W types, intensional type theory, homotopy type theory.
\end{keyword}
\end{frontmatter}
\section{Introduction}\label{intro}
It is well known that dependent type theories including Martin-Löf's type theory in \cite{PMTT} provide a foundational base both for a functional programming language (such as Haskell) and for constructive mathematics.

In this paper we show that a typical inductive data type
as that of well-founded trees has a topological counterpart in the point-free approach to topology, called formal topology.

Formal topology was introduced in \cite{S87} 
\iffalse
\cite{ somepoints}
\fi
to develop topology in a constructive and predicative foundation such as Martin-Löf's type theory in \cite{PMTT}.
With respect to the usual notion of ``locale'' it 
employs a severe distinction between a set of basic opens and a collection or class (not generally a set!) of formal opens
which are defined as fix-points of a closure operator on the set of basic opens.

Then, the need of building predicative and constructive examples of formal topologies, including the point-free topology of Dedekind real numbers,
inspired the advent of powerful inductive methods of topological generation put forward in \cite{CSSV03}. Since then, it was clear that some kind of well-founded set constructor was enough to formalize such a topological induction in Martin-Löf's type theory as shown in detail in \cite{silviobar}. Moreover, it was also underlined that the main difficulty in generating inductive topologies
reduces essentially to that of
generating inductive suplattices, named \textit{inductively generated basic covers}, because the structure of inductive frame can be easily instantiated as a special case of  inductive suplattice as shown in \cite{CSSV03}
and extensively explained in \cite{batsampre,cms13}.

Recently, in \cite{mmr21}, the Curry-Howard representation of intuitionistic connectives and quantifiers as types has been extended by giving a proof-relevant presentation of inductively generated basic covers within a  two-level extension of the Minimalist Foundation \cite{m09}. Moreover, combining Th.4.9 of \cite{mmr21} with Th.5.3 of \cite{RG94} it follows that a version of Martin-Löf's type theory with W-types has the same proof-theoretic strength as the one with inductively generated basic covers. This led to the following question: \textit{can we establish
directly in some version of Martin-Löf's type theory an equivalence between W-types and proof-relevant inductively generated basic covers?}

Inspired by the results in \cite{Petersson} and \cite{whyw}, in this paper, we show that over intensional Martin-Löf’s type theory, W-types and proof-relevant inductive basic covers can encode one another provided that:

\begin{enumerate}
\item function extensionality holds;
\end{enumerate}
or, alternatively, that
\begin{enumerate}
\setcounter{enumi}{1}
\item definitional $\eta$-equalities for $\Pi$-types, $\Sigma$-types and the singleton type $\mathsf{N_1}$ hold.
\end{enumerate}
Notice that the first hypothesis is satisfied by Homotopy Type Theory (see \cite{hottbook}), and the required $\eta$-equalities in the second set of hypotheses are usually implemented by default in the Agda proof-assistant.

The result we prove will actually involve two other type constructors, namely, dependent W-types, a generalisation of W-types introduced in \cite{depwf} and \textit{well-founded predicates}, a new type constructor which is a logical version  of dependent W-types for predicates. 
Finally, we will discuss the meaning of those type constructors and their relationship in the Minimalist Foundation.

\paragraph{Structure of the paper} The paper is structured as follows: in Section 2, we present the two type theories we will work in and the formal notions of propositional and definitional encoding for the two type constructors; in Section 3, we recall some basic notions of formal topology, leading to the definition of inductive basic cover; in Section 4, we recall W-types and dependent W-types, showing that they are mutually encodable in some type theories; in Section 5, we introduce the type constructor of well-founded predicates, showing in what type theories it is mutually encodable with respect to both W-types and inductive basic covers.

\paragraph{Contributions and related works}
It is well known that W-types can encode inductive datatypes. This was first proved in \cite{DyPe}, and recently in \cite{whyw} without the use of function extensionality.
The fact that W-types can encode also inductive families, represented by \textit{dependent} W-types, was first proved in extensional type theory in \cite{Petersson} and then in a categorical setting in \cite{w-typelccc}; in Homotopy Type Theory, a similar result concerning a slightly different generalization of W-types, namely \textit{indexed} W-types, has also been proved (see \cite{nlab:inductive_family}).

In \cite{CSSV03} and \cite{silviobar}, it was observed that the inductive generation of basic covers can be seen as a particular instance of the dependent W-type constructor.

Our main contribution is to formally show that, under some mild extensional hypotheses, extending an intensional dependent type theory with inductive generation of basic covers is equivalent to extend it with well-founded trees. To make this statement precise, we introduce the notion of \textit{encoding} between two type constructors. Finally, we discuss what equivalences survive in the Minimalist Foundation.

The results of this paper concerning the encoding of constructors in extensions of Martin-Löf's type theory have been checked in Agda; the formalization is available at \href{http://github.com/PietroSabelli/W-DW-IBC}{the second author's GitHub Page}.
%%Dybjer inductive types with W-types. Hugunin individua altre ipotesi estensionali con eta definizionale. The reduction of dependent W-types to W-types has been extensively... . The inductive generation of basic cover was always thought to be derivable from W-types.

\section{Preliminaries}
In this paper we work with two different intensional type theories. Namely, an intensional Martin-Löf's type theory\, and the intensional level of the Minimalist Foundation. We briefly recall here both.
\paragraph{Martin-Löf's type theory} We consider a version of intensional Martin-Löf's type theory $\mathbf{MLTT_0}$ with the following type constructors: the empty type $\mathsf{N_0}$, the unit type $\mathsf{N_1}$, dependent sums $\Sigma$, dependent products $\Pi$, identity types $\mathsf{Id}$, disjoint sums $+$, and a universe of small types $\mathsf{U_0}$ á la Russell closed under all the above type constructors. Inductive type constructors are defined as to allow elimination toward all (small and large) types; this will be true also for the inductive types introduced in the subsequent sections; in particular, this feature is due in order to recursively define predicates on a type. The intensionality of the theory means that judgmental equality is not reflected by propositional equality; moreover, since we perform a fine-grained analysis on the use of equality, we do not assume either $\eta$-equalities or function extensionality in our base theory; instead they will appear as additional hypothesis in statements.

\paragraph{Minimalist Foundation} The Minimalist Foundation is a two-level foundation conceived in \cite{mtt} and finalized in \cite{m09} equipped with an  intensional level $\mathbf{mTT}$ and an extensional level $\mathbf{emTT}$ and an interpretation of the latter in a (quotient model) of the first. Both levels are fomulated as dependent type theory \`a la  Martin-Löf. Here, we mainly work within the intensional level $\mathbf{mTT}$, in which there are four kinds of types: small propositions, propositions, sets  and collections (denoted respectively $\mathsf{prop_s}$, $\mathsf{prop}$, $\mathsf{set}$ and $\mathsf{col}$). Small propositions are both propositions and sets, and every type is a collection, as depicted in the following square of inclusions:
\begin{center}
    % https://tikzcd.yichuanshen.de/#N4Igdg9gJgpgziAXAbVABwnAlgFyxMJZABgBoBGAXVJADcBDAGwFcYkQAdDgW3pwAs4AM2BoAThDQB9OAF8Qs0uky58hFOQrU6TVuy68Bw4HBg55i5djwEim4toYs2iTjz6CRAYwiMLSkAxrNSIyBxonPVcDD2NxSQttGCgAc3giUCEJbiQyEBwIJE18+ixGdn4ICABrBQCsiBzEPIKkAGYaHFLy10qauszs9s7CxAAmTu6KqtrLEAam4tbxybLp-tlKWSA
\begin{tikzcd}
\mathsf{prop} \arrow[r, hook]                   & \mathsf{col}                 \\
\mathsf{prop_s} \arrow[r, hook] \arrow[u, hook] & \mathsf{set} \arrow[u, hook]
\end{tikzcd}
\end{center}
This allows one to differentiate on the one hand logical and mathematical entities, and on the other different degrees of complexity (corresponding to the usual small/large type distinction in a Martin-Löf's type theory\, with a universe). Small propositions include the falsum constant $\bot$ and propositional equalities $\mathsf{Id}(A,a,b)$ of two terms $a$ and $b$ in the same set $A$, and are closed under connectives $\wedge$, $\lor$, $\Rightarrow$ and quantification $\exists$, $\forall$ over sets; propositions include all small propositions, all propositional equalities $\mathsf{Id}(A,a,b)$ (even when $A$ is not a set), and are closed under all connectives and quantifiers (again, without any restriction on the domain of quantification); sets include the empty set $\mathsf{N}_0$, the singleton set $\mathsf{N}_1$, all small propositions (identified with the sets of their proofs) and are closed under dependent sum $\Sigma$, dependent product $\Pi$, disjoint sum $+$ and the list constructor $\mathsf{List}$. Finally, each set and each proposition is a collection; moreover, collections include a universe á la Russell of small propositions $\mathsf{prop_s}$ and are closed under dependent sum.
A crucial characteristic of the Minimalist Foundation is the fact that elimination rules of propositional constructors act only for propositions; for this reason the axioms of choice or of unique choice (and their rules) are  not a theorem as it is in Martin-L\"of type theory, since in general one cannot produce witnesses for existential statements (see \cite{choice}).

\paragraph{Notation}
In both theories we adopt the following notational conventions:
\begin{list}{-}{}
\item we use the notation of sequents to express type-judgments  instead of the original notation of natural deduction in \cite{PMTT,m09};
\item since we work in an intensional setting, we interpret the type $\mathcal{P}(A)$ of subsets of a small type $A$ as a setoid whose carrier is the type of predicates over $A$ following \cite{m09}, namely, the function space from $A$ to the universe of small propositions, rendered as the large type $A \to \mathsf{U_0}$ in Martin-Löf's type theory, and as the collection $A \to \mathsf{prop_s}$ in the Minimalist Foundation;
\item in a type-theoretical framework we need to be careful in distinguishing between typehood judgment, denoted here with the semicolon notation $a : A$, and the (propositional) relation of membership, denoted with the usual set-theoretical membership symbol $a \varepsilon V$, when $a : A$ and $V : \mathcal{P}(A)$; the same notation will be used when working with the intensional representation of subsets mentioned in the previous point;
\item when writing inference rules, the piece of context common to all the judgments appearing in an inference rule is omitted; moreover, when we give the rules of a type constructor, we interpret the formation rule's premises as parameters of the constructor; we then take them for granted in the premises of the other constructor's rules;
    \item we reserve the symbol  $\to$ \textit{(resp. $\times$)} as a shorthand for a non-dependent function space \textit{(resp. for non-dependent product spaces)}; moreover, when working in the Minimalist Foundation, we denote the implication connective with the arrow symbol $\Rightarrow$;
    \item when writing lambda abstractions, we omit to write the type of the abstraction;
    \item we write $a =_A b$ – or just $a = b$ when the type $A$ can be easily inferred from the context – as a shorthand for $\mathsf{Id}(A,a,b)$;
    \item we write $f(a)$ as a shorthand for $\mathsf{Ap}(f,a)$ \textit{(resp. $\mathsf{Ap_\forall}(f,a)$, $\mathsf{Ap_\Rightarrow}(f,a)$)} when $a : A$ and $f : (\Pi x : A)B(x)$ \textit{(resp. $f : (\forall x : A)B(x)$, $f : A \Rightarrow B$)}.
\end{list}

\paragraph{Encodings}

In the contexts of extensions of Martin-Löf's type theory and of the intensional level of the Minimalist Foundation, we define what it means for a type constructor to \textit{encode} another type constructor.

\begin{definition}
Let $\mathbf{T}$ be an extension of either $\mathbf{mTT}$, $\mathbf{emTT}$ or $\mathbf{MLTT_0}$, and let $\mathsf{C}$ and $\mathsf{D}$ be two type constructors of $\mathbf{T}$, intended as their corresponding sets of rules. We say that $\mathsf{C}$ \textit{definitionally encodes $\mathsf{D}$ in $\mathbf{T}$} if each new symbol appearing in $\mathsf{D}$ can be interpreted in $\mathbf{T} + \mathsf{C}$ in such a way that all the rules of $\mathsf{D}$ are valid under this interpretation. If $\mathbf{T}$ is an extension of $\mathbf{MLTT_0}$, we say that $\mathsf{C}$ \textit{propositionally encodes $\mathsf{D}$ in $\mathbf{T}$} if for each $A$ such that
$\mathbf{T} + \mathsf{D}$ derives $A \; \mathsf{type}$, there exist $B$ and $p$ such that $\mathbf{T} + \mathsf{C}$ derives $B \; \mathsf{type}$ and $\mathbf{T} + \mathsf{C} + \mathsf{D}$ derives $p : A \cong B$, where $A \cong B$ means $$(\Sigma f : A \to B)(\Sigma g : B \to A)((\Pi x : A)(g(f(x))=_A x) \times (\Pi x : B)(f(g(x))=_B x))$$

Finally, if it happens that in $\mathbf{T}$ both $\mathsf{C}$ definitionally \textit{(resp. propositionally)} encodes $\mathsf{D}$, and $\mathsf{D}$ definitionally \textit{(resp. propositionally)} encodes $\mathsf{C}$, we say that $\mathsf{C}$ and $\mathsf{D}$ are \textit{definitionally} (resp. propositionally) \textit{mutually encodable}.
\end{definition}

As already mentioned, we will consider extending a theory with the following axioms:
\begin{enumerate}
\item
function extensionality
\[
\frac
{
f : (\Pi x : A)B(x)
\quad
g : (\Pi x : A)B(x)
\quad
p : (\Pi x : A)(f(x) =_{B(x)} g(x))
}
{
\mathsf{funext}(f,g,p) : f =_{(\Pi x : A)B(x)} g
}    
\]
\item
$\eta$-equality for $\Pi$-types
\[
\frac
{
f : (\Pi x : A)B(x)
}
{
f = \lambda x . f(x) : (\Pi x : A)B(x)
}    
\]
\item
$\eta$-equality for $\Sigma$-types
\[
\frac
{
z : (\Sigma x : A)B(x)
}
{
z = \langle \mathsf{pr_1}(z) , \mathsf{pr_2}(z) \rangle : (\Sigma x : A)B(x)
}    
\]
\item
$\eta$-equality for $\mathsf{N_1}$
\[
\frac
{
z : \mathsf{N_1}
}
{
z = \star : \mathsf{N_1}
}    
\]
\end{enumerate}

\section{Inductive basic covers}
The name \say{Formal Topology} refers both to the study of topology in a constructive and predicative setting and to its main object of investigation: a point-free notion of topology whose definition, contrary to the classical one of topological space, avoids impredicative uses of the powerset. The core component of a formal topology is its underlying \textit{basic cover}, a predicative presentation of the topology's suplattice of open sets. We recall here its definition as originally appeared in \cite{batsampre}.

\begin{definition}
A \textit{basic cover relation} on a set $A$ is a binary relation $a \vartriangleleft V$ between elements $a : A$ and subsets $V : \mathcal{P}(A)$, such that the following two properties hold:
\begin{enumerate}
    \item \textit{(reflexivity)} if $a \varepsilon V$, then $a \vartriangleleft V$;
    \item \textit{(transitivity)} if $a \vartriangleleft U$, and $u \vartriangleleft V$ for each $u \varepsilon U$, then $a \vartriangleleft V$.
\end{enumerate}
\end{definition}

In this paper, we are specifically interested in a subclass of basic covers; namely the \textit{inductively generated} ones. The inductive generation of basic covers was devised by the authors in \cite{CSSV03} to have a convenient way for constructing formal topologies. Moreover, their method enjoys some desirable properties, such as the possibility of forming the product topology of two inductively generated formal topologies, which in general does not appear definable predicatively. We recall the notion of inductively generated basic cover, starting with the definition of axiom set.

\begin{definition}
An \textit{axiom set} consists of:
\begin{enumerate}
    \item a set $A$;
    \item an $A$-indexed family of sets $a: A\vdash I(a) \set\ $;
    \item for each $a : A$, an $I(a)$-indexed family of $A$'s subsets $  a : A,\, i : I(a)\vdash C(a,i) : \mathcal{P}(A)$.
\end{enumerate}
\end{definition}
The subsets family $C$ in the definition above is to be understood as a collection of axioms (hence the name axiom set) of the form $a \vartriangleleft C(a,i)$ for each $a : A$ and each $i : I(a)$. Given an axiom set, the basic cover inductively generated by it is the smallest basic cover that satisfies those axioms, formally:
\begin{definition}
A basic cover \textit{inductively generated} by an axiom set $A,I,C$ is a basic cover $\vartriangleleft$ on $A$ such that:
\begin{enumerate}
    \item $a \vartriangleleft C(a,i)$ holds for each $a : A$ and $i : I(a)$;
    \item if $\vartriangleleft'$ is another basic cover such that $a \vartriangleleft' C(a,i)$ holds for each $a : A$ and $i : I(a)$, then $a \vartriangleleft V$ implies $a \vartriangleleft' V$ for each $a : A$ and $V : \mathcal{P}(A)$.
\end{enumerate}
If a basic cover happens to be inductively generated by some axiom set, we say that it is an \textit{inductive basic cover}.
\end{definition}
Given an axiom set $A,I,C$, it is always possible to construct the basic cover inductively generated by it. It is the relation $\vartriangleleft_{I,C}$ obtained using the following generating rules:
\begin{enumerate}
    \item \textit{(reflexivity)} if $a \varepsilon V$, then $a \vartriangleleft_{I,C} V$;
    \item \textit{(infinity)} if, for some $a : A$ and $i : I(a)$, it holds that $b \vartriangleleft_{I,C} V$ for each $b \varepsilon C(a,i)$, then $a \vartriangleleft_{I,C} V$.
\end{enumerate}

In \cite{mmr21}, a new type constructor was added both in the Minimalist Foundation (\ref{ibc}) and, following the Curry-Howard interpretation, in \mltt\, (\ref{ibcMLTT}) for interpreting  the above generating method at an intensional level; there, examples of inductively basic covers were also defined, including the inductive topology of Dedekind real numbers and the inductive topology of the Cantor or Baire space. The type constructor follows the usual scheme for inductive types and in the following sections we will investigate how it compares to W-types and their dependent versions.
%-----------------------------------------
\section{Well-founded trees constructors}
The W-type constructor (also known as the type of well-orders or well-founded trees) was present in the first versions of \mltt\, (see \cite{ML84}), where it was used to constructively represent ordinals. Its rules in \mltt\, are listed in \ref{wMLTT}. For a set $A$ and an $A$-indexed family of sets $B$, the W-type $\mathsf{W}_{A,B}$ is generally understood set-theoretically as the set of (possibly infinitary) well-founded trees with nodes labelled by elements of a set $A$ and with a branching function given by $B$.

One of the main reasons for the importance of W-types is their ability, first observed by Dybjer in \cite{DyPe}, to encode common inductive types such as natural numbers and lists. Our goal is to show that this is also the case for the (Curry-Howard presentation of) inductive basic covers. %%Moreover, as Dybjer's results have been recently proved to hold also in an intensional setting, without necessarily require functions extensionality, we will also prove our result under such weaker hypotheses, taking advantage of the technique devised in \cite{whynotw}.
However, as soon as we confront the W-type constructor with that of inductive basic covers, we discover a substantial issue: the former produces just a set, while the latter produces a predicate, hence, by the propositions-as-types paradigm, a family of sets. This same limitation of the W-type constructor has been already addressed in \cite{Petersson}, where they proposed a generalization of W-types, called \textit{dependent W-types} (also known as \textit{indexed W-types}, or simply as \textit{trees}), capable of constructing \textit{families of mutually} inductive sets. This additional expressivity allows for example the construction of W-types in the setoid model without having to eliminate towards the universe, as shown in \cite{emmenegger}.

The rules of dependent W-types in \mltt\, are listed in \ref{dwMLTT}. Also dependent W-types can be interpreted as sets of well-founded trees, but their structure is more complex. The nodes of a dependent W-type are labelled by elements of a set $I$ as in the non-dependent case, however there is not just one branching associated to each label; instead, each label $i : I$ allows between a set $N(i)$ of possible branchings. Actually, an element $n : N(i)$ is just a label for the branching, the branching itself being given by a sets family $Br(i,n)$. Lastly, the trees of a dependent W-type satisfy a further property: the roots of each of its directed subtrees must be labelled according to an \textit{arity} function $ar(i,n) : Br(i,n) \to I$. We write a dependent W-type in symbol as $\mathsf{DW}_{Br,ar}$, omitting  the parameters $I$ and $N$ for readability.

Of course, W-types are a particular instance of dependent W-types. This can be formally seen by setting the parameter $I$ to be the unit type $\mathsf{N_1}$. Actually, and perhaps surprisingly, also the converse is true, namely dependent W-types can be reduced to non-dependent ones. In the following proposition we spell out our result, which is essentially a recasting in an intensional setting of the ideas contained in \cite{Petersson} and \cite{w-typelccc}.

\begin{proposition}\label{W-DW}

\begin{enumerate}
\item W-types propositionally encode dependent W-types in $\mathbf{MLTT_0}$ extended with function extensionality;
\item W-types definitionally encode dependent W-types in $\mathbf{MLTT_0}$ extended with $\eta$-equalities for $\Pi$-types and $\Sigma$-types.
\end{enumerate}

\iffalse
In $\mathbf{MLTT_0}$, dependent W-types can be encoded by W-types assuming either that function extensionality holds, or that $\eta$-equalities for $\Pi$-types and $\Sigma$-types hold. More precisely, given the parameters $I,N,Br,ar$ for constructing a dependent W-type, we can define a type family $i : I \vdash \mathsf{DW}'(i)$ such that:
\begin{enumerate}
\item if function extensionality holds, then $\mathsf{DW}'(i)$ is isomorphic to $\mathsf{DW}_{Br,ar}(i)$ for each $i:I$;
\item if $\eta$-equalities for $\Pi$-types and $\Sigma$-types hold, then $\mathsf{DW}'(i)$ satisfies the rules of the dependent W-types constructor.
\end{enumerate}
\fi
\end{proposition}
\begin{proof}
Assume to have the parameters
\begin{align*}
& I  : \mathsf{U_0} \\
& i : I \vdash N(i)  : \mathsf{U_0} \\
& i : I \,, n : N(i) \vdash Br(i,n)  : \mathsf{U_0} \\
& i : I \,, n : N(i) \vdash ar(i,n) : Br(i,n) \to I
\end{align*}
as in the premises of the formation rule of dependent W-types. Firstly, we construct a W-type $\mathsf{Free}$ of well-founded trees whose nodes are labelled by dependent pairs in $(\Sigma i : I)N(i)$ and whose branching function is given by the family $Br$ applied to the two projections. Formally, we are constructing the set
\[
\mathsf{Free} :\equiv \mathsf{W}_{(\Sigma i : I) N(i),Br(\mathsf{pr_1}(z),\mathsf{pr_2}(z))}
\]
This means that $\mathsf{Free}$ trees' nodes contain information both on how the nodes of the dependent W-type trees we are trying to simulate are labelled and on the branching options chosen for those nodes; all while respecting the same branching function.

Secondly, we impose that each node's label is in accordance with the arity parameter $ar$ by thinning out the set of $\mathsf{Free}$ trees with a $I$-indexed family of predicates $i : I \vdash \mathsf{Legal}(i): \mathsf{Free} \to \mathsf{U_0}$ which assert, for a $\mathsf{Free}$ tree, that its root's label has $i$ as its first component, and that each other node's first component is given by the arity function applied to its parent's labels and branch. The predicate is formally defined by recursion for any $i: I$, $ j: I$, $n: N(j)$ and $f: 
Br(j,n)\rightarrow \mathsf{Free}$ in the following way:
\[
\mathsf{Legal}(i, \mathsf{sup}(\langle j, n\rangle,f)):\, \equiv \,  (\Pi b : Br(j,n))\mathsf{Legal}(ar(j,n,b),f(b))\ \times\  (i =_I j)
\]
% ho scambiato i membri del prodotto definienti il legal perche' le coppie di legal scritte di seguito hanno esattamente  il tipo sopra
Note that, being $i$ fixed, the identity type $i =_I j$ has a unique proof propositionally, and hence it is contractible type according to \cite{hottbook}.

Then, our candidate for encoding the dependent W-type is the type family \[i : I \vdash \mathsf{DW'}(i) :\equiv (\Sigma w : \mathsf{Free})\mathsf{Legal}(i,w)\]
Now suppose to have $i : I, n : N(i)$ and $f : (\Pi b : Br(i,n))\mathsf{DW'}(ar(i,n,b))$, then we can straightforwardly define a constructor term of $\mathsf{DW'}(i)$ in the following way:
\[
\mathsf{dsup}'(i,n,f) :\equiv \langle \; \mathsf{sup}(\langle i , n\rangle, \lambda b. \mathsf{pr_1}(f(b))) \;,\;
                \langle \,    \lambda b .\mathsf{pr_2}(f(b))  \, , \mathsf{id}(i)\,\rangle \; \rangle
\]

For \textit{(i)}, we can now define by recursion a pair of functions $g_i : \mathsf{DW}_{Br,ar}(i) \to \mathsf{DW}'(i)$ and

$g^{-1}_i : \mathsf{DW}'(i) \to \mathsf{DW}_{Br,ar}(i)$ for each $i:I$ in the following way
\begin{align*}
g_i(\mathsf{dsup}(i,n,f)):\equiv & \;\mathsf{dsup'}(i, n, \lambda b . g_i(f(b))) \\
g^{-1}_i(\langle\mathsf{sup} (\langle i , n \rangle , f) , \langle l , \mathsf{id}(i) \rangle\rangle):\equiv & \;\mathsf{dsup}(i, n, \lambda b. g^{-1}_i(\langle f(b), l(b) \rangle))
\end{align*}

Using function extensionality, we can then check by induction that they are reciprocally inverses. 

\iffalse
We show the case for $g_i(g^{-1}_i(w)) =_{\mathsf{DW'}(i)} w$, the other case being analogous. On canonical terms we have the following chain of definitional equalities:
\begin{align*}
g_i(g^{-1}_i(\langle \mathsf{sup} (\langle i , n \rangle , f) , \langle l , \mathsf{id}(i) \rangle\rangle)) & = \qquad & \text{(by definition of $g^{-1}_i$)} \\
g_i(\mathsf{dsup}(i, n, \lambda b. g^{-1}_i(f(b), l(b)))) & = \qquad & \text{(by definition of $g_i$)} \\ 
\mathsf{dsup'}(i, n, \lambda b . g_i(g_i^{-1}\langle f(b), l(b) \rangle)) & = \qquad & \text{(by definition of $\mathsf{dsup}'$)} \\
\langle \; \mathsf{sup}(\langle i , n\rangle, \lambda b. \mathsf{pr_1}(g_i(g_i^{-1}\langle f(b), l(b) \rangle)) \;,\;
                \langle \, \lambda b .\mathsf{pr_2}(g_i(g_i^{-1}\langle f(b), l(b) \rangle)), \mathsf{id}(i)\, \rangle \; \rangle
\end{align*}
By induction, the above equality together with the following equalities between functions suffice to conclude:
\begin{align*}
\lambda b. \mathsf{pr_1}(g_i(g_i^{-1}\langle f(b), l(b) \rangle)) & = f \\
\lambda b. \mathsf{pr_2}(g_i(g_i^{-1}\langle f(b), l(b) \rangle)) & = l
\end{align*}
We can prove these function equalities using function extensionality:
\begin{align*}
\mathsf{pr_1}(g_i(g_i^{-1}(\langle f(b),l(b) \rangle))) & = \qquad & \text{(by inductive hypotesis)} \\
\mathsf{pr_1}(\langle f(b),l(b) \rangle) & = \qquad & \text{(by computation rule for $\Sigma$-types)} \\ 
f(b) &
\end{align*}
and similarly for the second one.
\fi

For \textit{(ii)}, we have already derived the formation and introduction rules of dependent W-types with $\mathsf{DW'}$ and $\mathsf{dsup'}$. Elimination and computation rules are a bit more involved, although they are not conceptually harder. Suppose to have, as in the premises of the elimination rule, a type family $i : I, w : \mathsf{DW}'(i) \vdash M(i,w)$ and a dependent term
\begin{align*}
& i : I , \\
& \quad n : N(i) , \\
& \quad\quad f : (\Pi b : Br(i,n))\mathsf{DW}'(ar(i,n,b)) , \\
& \quad\quad\quad h : (\Pi b : Br(i,n))M(ar(i,n,b),f(b)) \\
& \quad\quad\quad\quad\vdash d(i,n,f,h) : M(i,\mathsf{dsup}'(i,n,f))
\end{align*}
For the elimination rule we want to define a dependent term
\[
i : I,w:\mathsf{DW}'(i) \vdash \mathsf{El'_{DW}}(i,w,d) : M(i,w)
\]
satisfying  the following definitional equality to validate the computational rule
\begin{align*}
& i : I \\
& \quad n: N(i) \\ 
& \quad\quad f:(\Pi b:Br(i,n))\mathsf{DW}'(ar(i,n,b)) \\
& \quad\quad\quad \vdash \mathsf{El'_{DW}}(i,\mathsf{dsup}'(i,n,f),d) = d(i,n,f,\lambda b.\mathsf{El'_{DW}}(ar(i,n,f),f(b),d))
\end{align*}
The idea is to define the term $\mathsf{El'_{DW}}$ by recursion by mimicking the above requirement – so that the task of checking the computation rule will turn out to be trivial. The definition explicitly reads
\[
\mathsf{El'_{DW}}(i, \langle\mathsf{sup}(\langle i,n\rangle, f) , l , \mathsf{id}(i)\rangle, d) :\equiv d\;(\;i\;,\;n\;,\;\lambda b . \langle f(b) , l(b) \rangle\;,\;\lambda b.\mathsf{El'_{DW}}(ar(i,n, b), \langle f(b) , l(b) \rangle)\;)
\]
Above, we used the recursion principles of $\Sigma$-types, W-types and identity types, all at once. The long-but-routine calculations lie in checking that the given definition is well-typed by formulating it only with eliminator terms. In particular, it is in this step of defining $\mathsf{El'_{DW}}$ by recursion that $\eta$-equalities are needed to ensure that the calculations go through. We leave them to the assiduous reader or to the proof-checker.

\end{proof}

Taking advantage of the previous result, we will show the equivalence between W-types and inductive basic covers by proving, on the one hand, that dependent W-types can encode them, and, on the other, that (non-dependent) W-types can be encoded by them. However, we will not prove these two facts directly; instead, in the next section, we will introduce a new type constructor, called \textit{well-founded predicate}, to use as a bridge between (dependent) W-types and inductive basic covers. This intermediate step would not be strictly necessary for \mltt; nevertheless, aside from providing a clearer proof, its introduction is vital when we wish to keep apart the notions of set and proposition as it is done in the Minimalist Foundation. 

We close this section with a remark on the behaviour of W-types in the Minimalist Foundation which provides another, more technical reason why it is not convenient to work with them in such a framework.
\begin{remark}\label{now}
We know that in Martin-Löf's type theory, without a universe of sets, it is not possible to construct a family $x : A \vdash B(x) \set$ for which there exist $a, a' : A$ such that $B(a)$ is inhabited and $B(a')$ is (isomorphic to) the empty type. In the Minimalist Foundation, in which there is only a universe of (small) propositions, the same phenomenon occurs, with an additional subtlety: in fact, we can construct a non-always-inhabited, non-always-empty family of sets, e.g. $b : \mathsf{N_1} + \mathsf{N_1} \vdash \mathsf{El}_+(b,\bot,\top)$; however, no set of such a family seems to be provable isomorphic to the empty set $\mathsf{N_0}$, the main reason being that the falsum constant $\bot$ cannot eliminate towards sets. Consequently, in the Minimalist Foundation we cannot start the construction of a canonical element for a W-type (aside for those with a constantly-null branching function) because, according to the W-type introduction rule, that would require to construct a function towards the W-type itself, which is again a set.

Nonetheless, a suitable extension of the Minimalist Foundation (e.g. one equipped with a universe of sets) could make W-types (and their dependent version) actually usable; \iffalse for this reason, also in the Minimalist Foundation it makes sense to ask whether dependent W-types can be derived from W-types(the situation is similar to that in which one wants to prove properties of type constructors in a theory that does not necessarily have base types, from which the construction of actual types can start)\fi
a result analogous to point \textit{(ii)} of \ref{W-DW}  can be then shown for the intensional level of the Minimalist Foundation.
\end{remark}

%–––––––––––––––––––––––––––––––––––––––––––
\section{Well-founded predicates}
As mentioned in the previous section, the choice of representing predicates (such as inductive basic covers) through sets (such as well-founded trees) would not be consistent with the Minimalist Foundation's philosophy, which stipulates a strict distinction between propositions and sets. This brings us to introduce a new propositional constructor called \textit{well-founded predicate} to the Minimalist Foundation that allows for the inductive definition of predicates.

This new constructor stems as a logical counterpart for predicates of the notion of dependent well-founded tree, hence the name well-founded predicate. Analogously, a well-founded predicate can be interpreted as the set of \textit{proof} trees built using certain derivation rules defined by its parameters in the following way: $I$ is the set on which the well-founded predicate acts; for each element $i : I$, $N(i)$ is a set of names for inference rules whose conclusions state that the well-founded predicate holds for $i$; the premises of each rule $n : N(i)$ are (possibly infinite) instances of the well-founded predicate applied to some elements of $I$, and the predicate $R(i,n) : \mathcal{P}(I)$ tells precisely which ones. Schematically, we have the following entailments written in the language of the extensional level of the Minimalist Foundation:
\[
i : I , n : N(i) \vdash (\forall j \varepsilon R(i,n))\ \mathsf{WP}_R(j) \ \Rightarrow \ \mathsf{WP}_R(i)
\]
In particular, if $R(i,n)$ is the empty subset of $I$, the above entailment is interpreted as an axiom stating that the predicate holds for $i$; from there, the construction of a derivation can start (notice that we do not encounter the same problem for well-founded trees pointed out in \ref{now}, since now the well-founded predicate is a proposition towards which the falsum constant can eliminate).

The following representation property wraps up the above interpretation as a provable statement in the extensional level of the Minimalist Foundation:
$$
i : I \vdash \mathsf{WP}_{R}(i)\ \Leftrightarrow\ 
(\exists n : N(i))(\forall j \varepsilon R(i,n))\ \mathsf{WP}_{R}(j)
$$
It explicitly reads: \say{the well-founded predicate holds for $i$ if and only if there exists $n : N(i)$ such that all the premises $R(i,n)$ of the rules are satisfied by the well-founded predicate}.

In order to interpret this extensional predicate in the intensional level in a similar way to what is done in Proposition 2.7 of \cite{mmr21}  for the inductive cover predicate, it is enough to extend the intensional level of the Minimalist Foundation with the rules of the W-predicate constructors in the appendix \ref{wp}.

\iffalse
\begin{example}
Given a small binary relation $P(x,y) \props$ on a set $A$, we can define its reflexive, symmetric and transitive closure as the well-founded predicate constructed using the following parameters:
\begin{align*}
 I &:\equiv A \times A \\
 N(\langle x,y \rangle) &:\equiv P(x,y) + \mathsf{Id}(A,x,y) + \mathsf{N_1} + A \\
 R(\langle x,y \rangle,n) &:\equiv
\begin{cases}
\emptyset & \text{if } n : P(x,y)\\
\emptyset & \text{if } n : \mathsf{Id}(A,x,y)\\
\{\langle y,x \rangle\} & \text{if } n : \mathsf{N_1}\\
\{\langle x,n \rangle,\langle n,y \rangle\} & \text{if } n : A
\end{cases}
\end{align*}
\end{example}
\fi

Of course, the motivating example for introducing well-founded predicates is the inductive generation of basic covers. To formally see this, suppose to have an axiom set $A,I,C$ and a predicate $V : A \to \mathsf{prop_s}$; we can express the predicate $a \vartriangleleft_{I,C} V$ as a well-founded predicate over $A$ constructed using the following parameters:
\begin{align*}
N(a) &:\equiv V(a)+I(a) \\
R(a,n) &:\equiv
\begin{cases}
\emptyset & \text{if } n : V(a) \\
C(a,i) & \text{if } n : I(a)
\end{cases}
\end{align*}

If we also have at our disposal the basic inductive cover constructor, we can then form and prove in the theory the statement $a : A \vdash a \vartriangleleft_{I,C} V \Leftrightarrow \mathsf{WP}_{R}(a)$ asserting the equivalence between the two predicates. 

Actually, the basic inductive covers are a complete example of a well-founded predicate, in the sense that  their scope encompasses all the possibilities of the constructor. This is perhaps not surprising since, although interpreted differently, the rules of the former resemble very closely the ones of the latter, the only difference being the additional  canonical elements introduced by reflexivity.  To formally see this, suppose to have $I,N,R$ parameters for the well-founded predicates constructor; the statement $i : I \vdash \mathsf{WP}_{R}(i) \Leftrightarrow i \vartriangleleft_{N,R} \emptyset$ is then provable in the theory.

Of course, as for inductive basic covers, well-founded predicates can be defined also in \mltt\, by defining the elimination rule to act towards any dependent type (we spell out the rules in \ref{wpMLTT}), according to the propositions-as-types correspondence. Then, the previous results on the relationship between inductive basic covers and well-founded predicates based on equiprovability upgrade to the following statement.

\begin{proposition}\label{WP-IBC}

\begin{enumerate}
\item well-founded predicates are propositionally mutually encodable with inductive basic covers in $\mathbf{MLTT_0}$ extended with function extensionality;
\item well-founded predicates are definitionally mutually encodable inductive basic covers in $\mathbf{MLTT_0}$ and $\mathbf{mTT}$, both extended with $\eta$-equalities for $\Pi$-types and $\Sigma$-types, and also in  $\mathbf{emTT}$.
\end{enumerate}
\end{proposition}

\begin{proof}
Showing that inductive basic covers encode well-founded predicates in $\mathbf{MLTT_0}$ is a straightforward adaptation of the construction presented above in the case of the Minimalist Foundation.

The reverse is more involved, since it requires employing the technique presented in \cite{whyw} to correctly encode the reflexivity case of inductive basic covers. We show it only in the case of $\mathbf{MLTT_0}$, the case for $\mathbf{mTT}$ being entirely analogous.

We need to refine the construction with a predicate family

$a : A \vdash \mathsf{Canonical}(a) : \mathsf{WP}_R(a) \to \mathsf{U_0}$ defined by recursion in the following way (we leave implicit the dependence on $a$ of the predicate):
\[
\mathsf{Canonical}(\mathsf{ind}(a, n, f)) :\equiv
\begin{cases}
f =_{(\Pi j : I)(\mathsf{N_0} \to \mathsf{WP}_R(j))} \lambda j.\lambda z.\mathsf{El_{N_0}}(z) & \text{if } n : V(a) \\ 
(\Pi b : A)(\Pi r : R(a,n,b)) \mathsf{Canonical}(f(b, r)) & \text{if } n : I(a)
\end{cases}
\]

Reasoning analogously as in \ref{W-DW}, it is easy to check that the type family
\[
a : A \vdash (\Sigma w : \mathsf{WP}_R(a))\mathsf{Canonical}(w)
\]
is isomorphic to the inductive basic cover $a \vartriangleleft_{I,C} V$ (assuming function extensionality); or that it satisfies the rules of the inductive basic covers constructor (assuming $\eta$-equalities).
\end{proof}

We mentioned how well-founded predicates were added to the Minimalist Foundation as a natural propositional counterpart of dependent W-types. In \mltt\, we can show that the two constructors are indeed equivalent; to achieve this, we combine \ref{W-DW} with the following result.

\begin{proposition}\label{DW-WP}

\begin{enumerate}
\item Dependent W-types propositionally encode well-founded predicates in $\mathbf{MLTT_0}$ extended with function extensionality;
\item Dependent W-types definitionally encode well-founded predicates in $\mathbf{MLTT_0}$ extended with $\eta$-equality for $\Sigma$-types and $\Pi$-types;
\item Well-founded predicates propositionally encode dependent W-types in $\mathbf{MLTT_0}$ extended with function extensionality;
\item Well-founded predicates definitionally encode dependent W-types in $\mathbf{MLTT_0}$ extended with $\eta$-equality for $\Sigma$-types and the unit type $\mathsf{N_1}$.
\end{enumerate}
\end{proposition}

\begin{proof}

To show that dependent W-types encode well-founded predicates, assume to have parameters $I,N,R$ and recall that
\[R : (\Pi i : I)(\ N(i) \to (\ I \to \mathsf{U_0}\ ) \ )\] To construct the desired dependent W-type, we take for $I$ and $N$ the parameters with the same names; while, for each $i : I$ and $n : N(i)$, we define
\begin{align*}
Br(i,n) & :\equiv (\Sigma j : I) R(i,n,j)\\
ar(i,n) & :\equiv \mathsf{pr_1} : Br(i,n) \to I
\end{align*}
Without any further intricacies, the type family $i : I \vdash \mathsf{DW}_{(\Sigma j : I)R(i, n, j), \mathsf{pr_1}}$ can be proved to satisfy, under the chosen hypotheses, the corresponding statement as in propositions \ref{W-DW} and \ref{WP-IBC}.

To show that well-founded predicates encode W-types, assume to have the W-type parameters $A : \mathsf{U_0}$ and $B : A \to \mathsf{U_0}$. Firstly, we construct a well-founded predicate choosing for $I$ the unit type $\mathsf{N_1}$, for $N$ the constant type family at $A$, and for $R$ the type family $B$ depending only on the $A$'s indexing. We then instantiate the type family $x : \mathsf{N_1} \vdash \mathsf{WP}_{R}(x) : \mathsf{U_0}$ obtained in this way at the canonical element of the unit type $\star$, obtaining a type $\mathsf{W'} :\equiv \mathsf{WP}_{R}(\star) : \mathsf{U_0}$; it is straightforward to check that this type encodes the W-type constructor.
\end{proof}

\section{Conclusions}

We can combine propositions \ref{W-DW}, \ref{DW-WP} and \ref{WP-IBC}  to show, under the propositions-as-type paradigm, the essential equivalence of all the type constructors previously considered:

\begin{theorem}
The following type constructors are all propositionally mutually encodable in $\mathbf{MLTT_0}$ extended with function extensionality, and definitionally mutually encodable in $\mathbf{MLTT_0}$ extended with $\eta$-equalities for $\Pi$-types, $\Sigma$-types and $\mathsf{N_1}$:
\begin{enumerate}
\item W-types;
\item dependent W-types;
\item well-founded predicates;
\item inductive basic covers.
\end{enumerate}
\end{theorem}

In particular, we have the following corollaries: 

\begin{corollary}
\begin{enumerate}
\item W-types and inductive basic covers are propositionally mutually encodable in Homotopy Type Theory;
\item W-types and inductive basic covers are definitionally mutually encodable in the Agda implementation of intensional Martin-Löf's type theory.
\end{enumerate}
\end{corollary}

Finally, combining \ref{now} and \ref{WP-IBC}, we get a corresponding result for the intensional level of the Minimalist Foundation extended with $\eta$-equalities for $\Sigma$- and $\Pi$-types, and in $\mathbf{emTT}$, which already satisfies $\eta$-equalities. Clearly, the equivalence is no longer extended to all four type constructors since propositions and types are no longer identified:

\begin{theorem}
In $\mathbf{mTT}$ extended with $\eta$-equalities for $\Sigma$- and $\Pi$-types, and in $\mathbf{emTT}$ 
%(which already satisfies $\eta$-equalities),
the followings hold:
\begin{enumerate}
\item W-types and dependent W-types are definitionally mutually encodable;
\item basic inductive covers and well-founded predicates are definitionally mutually encodable.
\end{enumerate}
\end{theorem}

\paragraph{Conclusions and future work}
We have shown that in Martin-Löf’s type theory, an axiomatic treatment of inductive topology is not less powerful than the addition of W-types. In the Minimalist Foundation, this is true when W-types are replaced by well-founded predicates, given that the identification of all types as propositions is no longer valid.
\iffalse (for example, the falsum proposition does not represent the empty set). \fi
In future works, we hope to extend the comparison to the topological positivity relation defined in \cite{coind} with suitable coinductive predicates.

\iffalse
We have shown in Martin-Löf's type theory, an axiomatic treatment of inductive topology is not less powerful than the addition of W-types. In  future work we hope to clarify whether this is true also in the Minimalist Foundation as we expect in the light of remark \ref{now}.
\fi

\bibliography{main}{}
\bibliographystyle{entics}

\appendix \section{Rules of type constructors}

\paragraph{Rules for well-founded trees in Martin-Löf's type theory}
\label{wMLTT}

$$
\frac{A : \mathsf{U_0} \qquad
x : A \vdash B(x) : \mathsf{U_0}}{\mathsf{W}_{A,B} : \mathsf{U_0}}
\text{ F-}\textsf{W}
$$
\\
$$
\frac
{a : A \qquad f : B(a) \to \mathsf{W}_{A,B}}
{\mathsf{sup}(a,f) : \mathsf{W}_{A,B}}
\text{ I-}\textsf{W}
$$
\\
$$
\frac{
\begin{aligned}
& w : \mathsf{W}_{A,B} \vdash M(w) \set \qquad w : \mathsf{W}_{A,B}
\\
& a : A \,, f : B(a) \to \mathsf{W}_{A,B} \,, h : (\Pi b : B(a))M(f(b)) \vdash d(a,f,h) : M(\mathsf{sup}(a,f))
\end{aligned}}
{\mathsf{El_W}(w,d) : M(w)}
\text{ E-}\textsf{W}
$$
\\
$$
\frac{
\begin{aligned}
& w : \mathsf{W}_{A,B} \vdash M(w) \set
\qquad
a : A \qquad f : B(a) \to \mathsf{W}_{A,B}
\\
& a : A \,, f : B(a) \to \mathsf{W}_{A,B} \,, h : (\Pi b : B(a))M(f(b)) \vdash d(a,f,h) : M(\mathsf{sup}(a,f))
\end{aligned}}
{\mathsf{El_W}(\mathsf{sup}(a,f),d) = d(a,f,\lambda b.\mathsf{El_W}(f(b),d)) : M(\mathsf{sup}(a,f))}
\text{ C-}\textsf{W}
$$

\paragraph{Rules for dependent well-founded trees in Martin-Löf's type theory}
\label{dwMLTT}

$$
\frac{\begin{aligned}
& I  : \mathsf{U_0} \\
& i : I \vdash N(i)  : \mathsf{U_0} \\
& i : I \,, n : N(i) \vdash Br(i,n)  : \mathsf{U_0} \\
& i : I \,, n : N(i) \vdash ar(i,n) : Br(i,n) \to I
\end{aligned}
}{\mathsf{DW}_{Br,ar} : I \to \mathsf{U_0}}
\text{ F-}\textsf{DW}
$$
\\
$$
\frac{\displaystyle
i : I \qquad n : N(i) \qquad f : (\Pi b : Br(i,n))\mathsf{DW}_{Br,ar}(ar(i,n,b))
}{\mathsf{dsup}(i,n,f) : \mathsf{DW}_{Br,ar}(i)}\text{ I-}\textsf{DW}
$$
\\
$$
\frac{
\begin{aligned}
& i : I \,, w : \mathsf{DW}_{Br,ar}(i) \vdash M(i,w) : \mathsf{U_0}
\\
& i : I \qquad w : \mathsf{DW}_{Br,ar}(i) \\
& i : I , \\
& \quad n : N(i) , \\
& \quad\quad f : (\Pi b : Br(i,n))\mathsf{DW}_{Br,ar}(ar(i,n,b)) , \\
& \quad\quad\quad h : (\Pi b : Br(i,n))M(ar(i,n,b),f(b)) \\
& \quad\quad\quad\quad\vdash d(i,n,f,h) : M(i,\mathsf{dsup}(i,n,f))
\end{aligned}}
{\mathsf{El_{DW}}(i,w,d) : M(i,w)}
\text{ E-}\textsf{DW}
$$
\\
$$
\frac{
\begin{aligned}
& i : I \,, w : \mathsf{DW}_{Br,ar}(i) \vdash M(i,w) : \mathsf{U_0}
\\
& i : I \qquad n : N(i) \qquad f : (\Pi_ b : Br(i,n))\mathsf{DW}_{Br,ar}(ar(i,n,b)) \\
& i : I , \\
& \quad n : N(i) , \\
& \quad\quad f : (\Pi b : Br(i,n))\mathsf{DW}_{Br,ar}(ar(i,n,b)) , \\
& \quad\quad\quad h : (\Pi b : Br(i,n))M(ar(i,n,b),f(b)) \\
& \quad\quad\quad\quad\vdash d(i,n,f,h) : M(i,\mathsf{dsup}(i,n,f))
\end{aligned}}
{\mathsf{El_{DW}}(i,\mathsf{dsup}(i,n,f),d) = d(i,n,f,\lambda b.\mathsf{El_{DW}}(ar(i,n,b),f(b),d)) : M(i,\mathsf{dsup}(i,n,f))}
\text{ C-}\textsf{DW}
$$

\paragraph{Rules for inductive basic covers in $\mathbf{mTT}$}
\label{ibc}

$$
\frac{\begin{aligned}
& A \set \\
& a : A \vdash I(a) \set \\
& a : A \,, i : I(a) \vdash C(a,i) : A \to \mathsf{prop_s} \\
& V : A \to \mathsf{prop_s}
\end{aligned}}{a : A \vdash a \vartriangleleft_{I,C} V \props}
\;\text{ F -}\vartriangleleft
$$
\\
$$
\frac{a : A \qquad r : V(a)}
{\mathsf{rf}(a,r) : a \vartriangleleft_{I,C} V}
\;\text{I}_{\mathsf{rf}}\text{ -}\vartriangleleft
$$
\\
$$
\frac{a : A \qquad i : I(a) \qquad r : (\forall b : A)(C(a,i,b) \Rightarrow b \vartriangleleft_{I,C} V)}
{\mathsf{tr}(a,i,r) : a \vartriangleleft_{I,C} V}
\;\text{I}_{\mathsf{tr}}\text{ -}\vartriangleleft
$$
\\
$$
\frac{
\begin{aligned}
& a : A \vdash P(a) \propa \qquad a : A \qquad p : a \vartriangleleft_{I,C} V
\\
 & a : A \,, r : V(a) \vdash q_1(a,r) : P(a)
 \\
 & a : A\,, i : I(a)\,, s : (\forall b : A)(C(a,i,b) \Rightarrow P(b)) \vdash q_2(a,i,s) : P(a)
   \end{aligned}}{\mathsf{El}_{\vartriangleleft}(p,q_1,q_2) : P(a)}
\;\text{ E -}\vartriangleleft
$$
\\
$$
\frac{
\begin{aligned}
& a : A \vdash P(a) \propa \qquad a : A \qquad r : V(a)
\\
 & a : A \,, r : V(a) \vdash q_1(a,r) : P(a)
 \\
 & a : A\,, i : I(a)\,, s : (\forall b : A)(C(a,i,b) \Rightarrow P(b)) \vdash q_2(a,i,s) : P(a)
   \end{aligned}}{\mathsf{El}_{\vartriangleleft}(\mathsf{rf}(a,r),q_1,q_2) = q_1(a,r) : P(a)}
\;\text{C}_{\mathsf{rf}}\text{ -}\vartriangleleft
$$
\\
$$
\frac{
\begin{aligned}
& a : A \vdash P(a) \propa \\
& a : A \qquad i : I(a) \qquad r : (\forall b : A)(C(a,i,b) \Rightarrow b \vartriangleleft_{I,C} V) 
\\
 & a : A \,, r : V(a) \vdash q_1(a,r) : P(a)
 \\
 & a : A\,, i : I(a)\,, s : (\forall b : A)(C(a,i,b) \Rightarrow P(b)) \vdash q_2(a,i,s) : P(a)
   \end{aligned}}{\mathsf{El}_{\vartriangleleft}(\mathsf{tr}(a,i,r),q_1,q_2) = q_2(a,i,\lambda_{\forall} b.\lambda_{\Rightarrow} t.\mathsf{El}_{\vartriangleleft}(r(b,t),q_1,q_2)) : P(a)}
\;\text{C}_{\mathsf{tr}}\text{ -}\vartriangleleft
$$

\paragraph{Rules for inductive basic covers in Martin-Löf's type theory}
\label{ibcMLTT}

$$
\frac{\begin{aligned}
& A : \mathsf{U_0} \\
& a : A \vdash I(a) : \mathsf{U_0} \\
& a : A \,, i : I(a) \vdash C(a,i) : A \to \mathsf{U_0} \\
& V : A \to \mathsf{U_0}
\end{aligned}}{a : A \vdash a \vartriangleleft_{I,C} V : \mathsf{U_0}}
\;\text{ F-}\textsf{$\vartriangleleft$}
$$
\\

$$
\frac{a : A \qquad r : V(a)}
{\mathsf{rf}(a,r) : a \vartriangleleft_{I,C} V}
\;\textsf{I\textsubscript{rf} -$\vartriangleleft$}
$$
\\

$$
\frac{a : A \qquad i : I(a) \qquad r : (\Pi b : A)(C(a,i,b) \to b \vartriangleleft_{I,C} V)}
{\mathsf{tr}(a,i,r) : a \vartriangleleft_{I,C} V}
\;\text{I - }\vartriangleleft
$$
\\

$$
\frac{
\begin{aligned}
& a : A \,, p : a \vartriangleleft_{I,C} V \vdash M(a,p) \set \\
& a : A \qquad p : a \vartriangleleft_{I,C} V \\
& a : A \,, r : V(a) \vdash q_1(a,r) : M(a,\mathsf{rf}(a,r)) \\
& a : A\,, i : I(a)\,, \\
& \quad r : (\Pi b : A)(C(a,i,b) \to b \vartriangleleft_{I,C} V) \,, \\
& \quad \quad h : (\Pi b : A)(\Pi s : C(a, i, b))M(b,r(b,s)) \\
& \quad \quad \quad \vdash q_2(a,i,r,h) : M(a,\mathsf{tr}(a,i,r))
   \end{aligned}}{\mathsf{El}_{\vartriangleleft}(p,q_1,q_2) : M(a,p)}
\;\text{ E -}\vartriangleleft
$$
\\

$$
\frac{
\begin{aligned}
& a : A \,, p : a \vartriangleleft_{I,C} V \vdash M(a,p) \set \\
& a : A \qquad r : V(a)  \\
& a : A \,, r : V(a) \vdash q_1(a,r) : M(a,\mathsf{rf}(a,r)) \\
& a : A\,, i : I(a)\,, \\
& \quad r : (\Pi b : A)(C(a,i,b) \to b \vartriangleleft_{I,C} V) \,, \\
& \quad \quad h : (\Pi b : A)(\Pi s : C(a, i, b))M(b,r(b,s)) \\
& \quad \quad \quad \vdash q_2(a,i,r,h) : M(a,\mathsf{tr}(a,i,r))
   \end{aligned}}{\mathsf{El}_{\vartriangleleft}(\mathsf{rf}(a,r),q_1,q_2) = q_1(a,r) : M(a,\mathsf{rf}(a,r))}
\;\text{ C}_{\mathsf{rf}}\text{ -}\vartriangleleft
$$
\\

$$
\frac{
\begin{aligned}
& a : A \,, p : a \vartriangleleft_{I,C} V \vdash M(a,p) \set \\
& a : A \qquad i : N(a) \qquad r : (\Pi b : A)(C(a,i,b) \to b \vartriangleleft_{I,C} V) \\
& a : A \,, r : V(a) \vdash q_1(a,r) : M(a,\mathsf{rf}(a,r)) \\
& a : A\,, i : I(a)\,, \\
& \quad r : (\Pi b : A)(C(a,i,b) \to b \vartriangleleft_{I,C} V) \,, \\
& \quad \quad h : (\Pi b : A)(\Pi s : C(a, i, b))M(b,r(b,s)) \\
& \quad \quad \quad \vdash q_2(a,i,r,h) : M(a,\mathsf{tr}(a,i,r))
   \end{aligned}}{\mathsf{El}_{\vartriangleleft}(\mathsf{tr}(a,i,r),q_1,q_2) = q_2(a, i, r, \lambda b.\lambda p. \mathsf{El}_{\vartriangleleft}(r(b, p), q_1, q_2) : M(a,\mathsf{tr}(a,i,r))}
\;\text{ C}_{\mathsf{tr}}\text{ -}\vartriangleleft
$$

\paragraph{Rules for well-founded predicates in $\mathbf{mTT}$}\label{wp}

$$
\frac{\begin{aligned}
& I \set \\
& i : I \vdash N(i) \set \\
& i : I \,, n : N(i) \vdash R(i,n) : I \to \mathsf{prop_s}
\end{aligned}}{
i : I \vdash \mathsf{WP}_{R}(i) \props
}
\text{ F-}\textsf{WP}
$$
\\
$$
\frac{i : I \qquad n : N(i) \qquad r : (\forall j : I)(R(i,n,j) \Rightarrow \mathsf{WP}_{R}(j))}
{\mathsf{ind}(i,n,r) : \mathsf{WP}_{R}(i)}
\;\text{ I-}\textsf{WP}
$$
\\
$$
\frac{
\begin{aligned}
& i : I \vdash P(i) \propa \qquad i : I \qquad p : \mathsf{WP}_{R}(i)
 \\
 & i : I, n : N(i), s : (\forall j : I)(R(i,n,j) \Rightarrow P(j)) \vdash c(i,n,s) : P(i)
   \end{aligned}}{\mathsf{El}_{\mathsf{WP}}(i,p,c) : P(i)}
 \text{ E-}\textsf{WP}
$$
\\
$$
\frac{
\begin{aligned}
& i : I \vdash P(i) \propa \\
& i : I \qquad n : N(i) \qquad r : (\forall j : I)(R(i,n,j) \Rightarrow \mathsf{WP}_{R}(j)) 
 \\
 & i : I, n : N(i), s : (\forall j : I)(R(i,n,j) \Rightarrow P(j))\vdash c(i,n,s) : P(i)
   \end{aligned}}
{\mathsf{El}_{\mathsf{WP}}(i,\mathsf{ind}(i,n,r),c) = c(i,n,\lambda_{\forall} j.\lambda_{\Rightarrow} s.\mathsf{El}_{\mathsf{WP}}(j,r(j,s),c)) : P(i)}
 \text{ C-}\textsf{WP}
$$

\paragraph{Rules for well-founded predicates in Martin-Löf's type theory}
\label{wpMLTT}

$$
\frac{\begin{aligned}
& I : \mathsf{U_0} \\
& i : I \vdash N(i) : \mathsf{U_0} \\
& i : I \,, n : N(i) \vdash R(i,n) : I \to \mathsf{U_0}
\end{aligned}}{
i : I \vdash \mathsf{WP}_{R}(i) : \mathsf{U_0}
}
\text{ F-}\textsf{WP}
$$
\\
$$
\frac{i : I \qquad n : N(i) \qquad f : (\Pi j : I)(R(i,n,j) \to \mathsf{WP}_R(j))}
{\mathsf{ind}(i,n,f) : \mathsf{WP}_{R}(i)}
\;\text{ I-}\textsf{WP}
$$
\\
$$
\frac{
\begin{aligned}
& i : I , w : \mathsf{WP}_{R}(i) \vdash M(i,w) \set \qquad i : I \qquad w : \mathsf{WP}_{R}(i)
 \\
& i: I, \\
& \quad n : N(i), \\
& \quad\quad f : (\Pi j : I)(R(i,n,j) \to \mathsf{WP}_R(j)) , \\
& \quad\quad\quad h : (\Pi j : I)(\Pi r : R(i,n,j))M(j, f(j, r)) \\
& \quad\quad\quad\quad\vdash c(i,n,f,h) : M(i,\mathsf{ind}(i,n,f))
   \end{aligned}}{\mathsf{El}_{\mathsf{WP}}(i,w,c) : M(i,w)}
 \text{ E-}\textsf{WP}
$$
\\
$$
\frac{
\begin{aligned}
& i : I , w : \mathsf{WP}_{R}(i) \vdash M(i,w) \set \\
&i : I \qquad n : N(i) \qquad f : (\Pi j : I)(R(i,n,j) \to \mathsf{WP}_R(j))\\
& i: I, \\
& \quad n : N(i), \\
& \quad\quad f : (\Pi j : I)(R(i,n,j) \to \mathsf{WP}_R(j)) , \\
& \quad\quad\quad h : (\Pi j : I)(\Pi r : R(i,n,j))M(j, f(j, r)) \\
& \quad\quad\quad\quad\vdash c(i,n,f,h) : M(i,\mathsf{ind}(i,n,f))
   \end{aligned}}
{\mathsf{El_{WP}}(i, \mathsf{ind}(i, n, f), c) = c(i, n, f, \lambda j.\lambda r. \mathsf{El_{WP}}(c, j, f(j, r)) : M(i,\mathsf{ind}(i,n,f))}
 \text{ C-}\textsf{WP}
$$
\end{document}